\documentclass[runningheads]{llncs}
\usepackage{cite}
\usepackage{amsmath,amssymb,amsfonts}
\usepackage{algorithmic}
\usepackage{graphicx}
\usepackage{textcomp}
\usepackage{xcolor}
\def\BibTeX{{\rm B\kern-.05em{\sc i\kern-.025em b}\kern-.08em
    T\kern-.1667em\lower.7ex\hbox{E}\kern-.125emX}}

\usepackage[utf8]{inputenc}
\usepackage{epsfig}
\usepackage{psfrag}
\usepackage{comment}
\usepackage[switch]{lineno}
\usepackage{contour}
\usepackage{balance}

\usepackage{amsfonts,amssymb,amsmath,latexsym,ae,aecompl,bbm,algorithm,amsthm}
\usepackage{color}

\newtheorem{observation}{Observation}

\begin{document}

\title{Byzantine Dispersion on Graphs \thanks{A preliminary version of this work~\cite{MMM21} appeared in the proceedings of the 35th IEEE International Parallel \& Distributed Processing Symposium (IPDPS 2021).}}
\author{Anisur Rahaman Molla\inst{1}\orcidID{0000-0002-1537-3462}\thanks{The work of A. R. Molla was supported, in part, by DST INSPIRE Faculty Research Grant DST/INSPIRE/04/2015/002801, Govt. of India and ISI DCSW/TAC Project (file number E5412).} \and
Kaushik Mondal\inst{2}\orcidID{0000-0002-9606-9293} \and
William K. Moses Jr.\inst{3}\orcidID{0000-0002-4533-7593}\thanks{Part of the work was done while W. K. Moses Jr. was a postdoc at the Faculty of Industrial Engineering and Management at the Technion - Israel Institute of Technology in Haifa, Israel. The work of W. K. Moses Jr. was supported in part by a Technion fellowship and in part by NSF grants, CCF1540512, IIS-1633720, CCF-1717075, and BSF grant 2016419.}}

\authorrunning{Molla et al.}
\institute{Computer and Communication Sciences, Indian Statistical Institute, Kolkata, India.\\
\email{molla@isical.ac.in} \and
Dept. of Mathematics, Indian Institute of Technology Ropar, Ropar, India.\\
\email{kaushik.mondal@iitrpr.ac.in} \and
Dept. of Computer Science, University of Houston, Houston, USA.\\
\email{wkmjr3@gmail.com}}

\maketitle

\begin{abstract}
This paper considers the problem of Byzantine dispersion and extends previous work along several parameters. The problem of Byzantine dispersion asks: given $n$ robots, up to $f$ of which are Byzantine, initially placed arbitrarily on an $n$ node anonymous graph, design a terminating algorithm to be run by the robots such that they eventually reach a configuration where each node has at most one non-Byzantine robot on it.

Previous work solved this problem for rings and tolerated up to $n-1$ Byzantine robots. In this paper, we investigate the problem on more general graphs. We first develop an algorithm that tolerates up to $n-1$ Byzantine robots and works for a more general class of graphs.

We then develop an algorithm that works for \emph{any} graph but tolerates a lesser number of Byzantine robots.

We subsequently turn our focus to the strength of the Byzantine robots. Previous work considers only ``weak" Byzantine robots that cannot fake their IDs. We develop an algorithm that solves the problem when Byzantine robots are not weak and can fake IDs.

Finally, we study the situation where the number of the robots is not $n$ but some $k$. We show that in such a scenario, the number of Byzantine robots that can be tolerated is severely restricted. Specifically, we show that it is impossible to deterministically solve Byzantine dispersion when $\lceil k/n \rceil > \lceil (k-f)/n \rceil$.
\end{abstract}

\section{Introduction}
\label{sec:intro}

The questions of {\em what can be computed} by independent computational entities working together and {\em how fast can it be computed} serve as a high level motivation for research that spans a broad variety of fields such as population protocols~\cite{AADFP06}, mobile robots~\cite{PRT11}, and programmable matter~\cite{DHRS19}
among others. Mobile robots on graphs, as a paradigm, restricts this question to situations where the entities are restricted in their movement and communication capabilities. Within this area, problems that are studied take on the form of either having the robots work together to find something in the graph (e.g., exploration~\cite{Bampas:2009,Cohen:2008,Das13,Dereniowski:2015,Fraigniaud:2005,MencPU17}, treasure hunting~\cite{MP15})
or form a certain configuration (e.g., gathering~\cite{CFPS12,CP02,DKLHPW11,P07}, scattering~\cite{Barriere2009,ElorB11,Poudel18,Shibata:2016}, pattern formation \cite{SY99},
convergence \cite{CP04}).

Dispersion is one such problem of the latter category. Introduced in this setting by Augustine and Moses Jr.~\cite{Augustine:2018}, it asks the following question. Given $n$ robots initially placed arbitrarily on an $n$ node graph, devise an algorithm such that the robots reach a configuration where exactly one robot is present on each node. This problem can be used as an abstraction for the large swath of problems where computational entities must share resources with the constraint that sharing the same resource is much more expensive than searching for an unused resource. The original paper looked at the trade-offs between time taken to reach this configuration and the memory required by each robot. Subsequent papers~\cite{Kshemkalyani,KMS2019,KMS2020-ICDCN,KMS2020,tamc19} have expanded the scope of this problem. However, only the recent papers of Molla et al.~\cite{MMM20-byzantine,MMM21-TCS} expands the scope of this problem to capture the issue of Byzantine faults. More specifically, how can the problem of dispersion be adapted and solved when some fraction of the robots may act in a Byzantine manner. In the real world, where systems must be fault tolerant and errors are the norm instead of the exception, more work must be done to understand and deal with them.

The previous work sought to study this problem on a ring and proposed time and memory optimal algorithms for that setting. In this paper, we expand on that work and broaden it in a variety of ways.

\subsection{Model}\label{subsec:model}
Consider a graph with $n$ nodes and $m$ edges. The nodes are unlabeled and have ports leading to other nodes. The ports of a node have unique labels in $[1,\delta]$ where $\delta$ is the degree of the node. Note that an edge between adjacent nodes may have different port numbers assigned to it by the corresponding nodes.

Each robot has a unique identifier assigned to it from the range $[1,n^c]$, where $c>1$ is a constant, unless otherwise stated. Two robots co-located on the same node can communicate with each other via messages.\footnote{In contrast to the model used in~\cite{MMM20-byzantine}, we allow robots to communicate with each other via messages instead of exposed memory. We believe this allows for a cleaner description of the model as well as behavior that is more closely related to the original notions of weak and strong Byzantine robots from~\cite{DPP14}, to be described later.} If a robot moves from one node to an adjacent node, it is aware of both port numbers assigned to the edge through which it passed. Each robot knows the value of $n$.

We use the notion of weak and strong Byzantine robots from~\cite{DPP14}. A (weak or strong) Byzantine robot may behave maliciously and arbitrarily, i.e., it may share wrong information, perform moves that are deviations from the algorithm, etc. A weak Byzantine robot is restricted in that it cannot fake its ID, i.e., its ID is attached to messages sent by it and and cannot be altered. A strong Byzantine robot has no such restriction on it. Among the $n$ robots, up to $f$ of them are considered to be Byzantine. Robots do not know the value of $f$ except for one algorithm in Section~\ref{sec:alg-strong-byz-robots}.

We consider a synchronous system, where in each round a robot performs the following tasks in order: (i) Robots that are co-located at the same node may communicate with each other.  Robots may perform some local computation and decide to move away from the node via a port. (ii) Robots perform a movement, if any, to an adjacent node. We also make use of the idea of \textit{breaking a round into sub-rounds} from~\cite{tamc19} when designing the procedure {\sc Dispersion-Using-Map} (see Section~\ref{sec:dispersion-using-map}). The concept of breaking a round into sub-rounds is just to divide a round into smaller fragments of time, i.e., sub-rounds. All the sub-rounds of a round may involve local computation and communication (task (i)). Only at the end of the round (which may be considered the final sub-round of the round), can robots move (task (ii)). This concept is needed when we need a little bit of synchronization between different instances of the local computation and communication. We assume that all robots are initially awake and can engage in the algorithms from the beginning.


We formally state the problem of Byzantine dispersion below, slightly modified from the definition in~\cite{MMM20-byzantine,MMM21-TCS}.
\begin{definition}[Byzantine dispersion]
\label{def:byzantine-dispersion}
Given $n$ robots, up to $f$ of which are Byzantine, initially placed arbitrarily on a graph of $n$ nodes, the non-Byzantine robots must re-position themselves autonomously to reach a configuration where each node has at most one non-Byzantine robot on it and subsequently terminate.
\end{definition}
\subsection{Our Contributions}\label{subsec:our-contrib}
We first present an algorithm to solve Byzantine dispersion that tolerates up to $n-1$ Byzantine robots and runs in $polynomial(n)$ rounds where robots start in an arbitrary configuration in Section~\ref{sec:alg-upto-n-1}. However, the algorithm only works on the class of graphs for which the quotient graph of the graph is isomorphic to the original graph.

\begin{sloppypar}
In Section~\ref{sec:alg-general-graphs}, we develop increasingly faster algorithms for more general graphs at the expense of handling less Byzantine robots. We first develop an algorithm that runs on any graph when robots start in an arbitrary configuration that takes $O(n^4 |\Lambda_{good}| X(n))$ rounds, where $\Lambda_{good}$ is the length of the largest ID among all non-Byzantine robots and $X(n)$ is the number of rounds required to explore any graph of $n$ nodes, but tolerates up to $\lfloor n/2 - 1 \rfloor$ Byzantine robots. From~\cite{AKLLR79} and~\cite{TZ14}, we know that $X(n) = \Tilde{O}(n^5)$, so the algorithm takes $\Tilde{O}(n^9)$.\footnote{In fact, when more information of the graph is known, $X(n)$ is even smaller. If the maximum degree of the graph is $d$, then $X(n) = \Tilde{O}(d^2 n^3)$ and if the graph is a simple $d$-regular graph, then $X(n) = \Tilde{O}(dn^3)$.}  We show that this algorithm takes only $O(n^4)$ rounds when robots start in a {\em gathered configuration}, i.e., all robots start at the same node. We then show how to speed up the algorithm to take $O(n^3)$ rounds for an initially gathered configuration when the number of Byzantine robots is restricted to at most $\lfloor n/3 -1 \rfloor$. We then show that when the number of Byzantine robots is highly restricted to at most $O(\sqrt{n})$, we can develop an algorithm to solve Byzantine dispersion from an arbitrary initial configuration in $O((f +  |\Lambda_{all}|) X(n))$ rounds, where $\Lambda_{all}$ is the length of the largest ID among all robots. Notice that $O((f +  |\Lambda_{all}|) X(n)) = \Tilde{O}(n^5 \sqrt{n})$.
\end{sloppypar}

All the above results are only for weak Byzantine robots. In Section~\ref{sec:alg-strong-byz-robots}, we develop an algorithm that solves Byzantine dispersion in the presence of at most $\lfloor n/4 - 1 \rfloor$ strong Byzantine robots in $exponential(n)$ rounds when robots are initially arbitrarily located. We show that when robots start from a gathered configuration, the algorithm takes $O(n^3)$ rounds.

Finally, in Section~\ref{sec:impossibility-result}, we show that when the total number of robots $k$ is different from the number of nodes in the graph, it is impossible to deterministically solve Byzantine dispersion when there are too many Byzantine robots, i.e., $\lceil k/n \rceil > \lceil (k-f)/n \rceil$, even if those Byzantine robots are weak. Our results are summarized in Table~\ref{table:results}.

\begin{table*}[ht]
	\caption{Our results for Byzantine dispersion of $n$ robots on an $n$ node graph in the presence of at most $f$ Byzantine robots. The third column gives the time complexity of the result. The fourth column indicates whether all robots start at the same node (gathered) or from an arbitrary configuration. The fifth column indicates up to how many Byzantine robots the algorithm can tolerate. The final column indicates whether the algorithm can handle strong Byzantine robots. Note that $\Lambda_{good}$ is the length of the largest ID among non-Byzantine robots, $\Lambda_{all}$ is the length of the largest ID among all robots, and $X(n)$ is the number of rounds required to explore any graph of $n$ nodes.}
\centering 
		\resizebox{\columnwidth}{!}{
	\begin{tabular}{|c|c|c|c|c|c|}
		\hline
		Serial & Theorem & Running Time &  Starting & Byzantine & Handles Strong\\
		 No. & No. & (in rounds) &   Configuration & Tolerance &  Byzantine Robots  \\
		\hline
		\hline
		1* &  1 & $polynomial(n)$ &  Arbitrary & $n-1$ & No \\
		\hline
		2$^\maltese$ & 2 & $O(n^4 |\Lambda_{good}| X(n))$  & Arbitrary & $\lfloor n/2 -1 \rfloor$ & No \\
		\hline
		3$^\diamond$ & 5 & $O((f +  |\Lambda_{all}|) X(n))$ &  Arbitrary & $O(\sqrt{n})$ & No \\
		\hline
		\hline
		4 & 3 & $O(n^4)$ &  Gathered & $\lfloor n/2 - 1 \rfloor$ & No\\
		\hline
		5 & 4 & $O(n^3)$ &  Gathered & $\lfloor n/3 - 1 \rfloor$ & No\\
		\hline
		\hline
		6$^\dagger$ & 7 & $exponential(n)$  & Arbitrary & $\lfloor n/4 - 1\rfloor$ &  Yes \\
		\hline
		7 & 6 & $O(n^3)$  & Gathered & $\lfloor n/4 - 1\rfloor$ & Yes \\
		\hline
		\hline
		\multicolumn{6}{|l|}{*This result holds only for those graphs where the quotient graph is isomorphic to the original graph.} \\
		\multicolumn{6}{|l|}{$^\maltese$Since $|\Lambda_{good}| = O(\log n)$ and $X(n) = \Tilde{O}(n^5)$ (see~\cite{AKLLR79,TZ14}), $O(n^4 |\Lambda_{good}| X(n)) = \Tilde{O}(n^9)$.} \\
		\multicolumn{6}{|l|}{$^\diamond$Since $|\Lambda_{all}| = O(\log n)$,$f=O(\sqrt{n})$, and $X(n) = \Tilde{O}(n^5)$ (see~\cite{AKLLR79,TZ14}), $O((f +  |\Lambda_{good}|) X(n)) = \Tilde{O}(n^5 \sqrt{n})$.} \\
		\multicolumn{6}{|l|}{$^\dagger$This result requires the robots to know the value of $f$.} \\
		\hline
	\end{tabular}
		}
	\label{table:results}
\end{table*}

\subsection{Technical Difficulties and High-Level Ideas}\label{subsec:tech-diff}
It is relatively simple to have robots work together to solve dispersion when all robots are non-Byzantine. One of the big issues when robots are Byzantine is that they may lie when they communicate with other robots. How then can one handle such robots?

Let us look a little deeper into the issue. Where does honest communication help robots solve the problem? Honesty is needed so that robots can learn which nodes are available for them to settle down at. But what if we approach this issue in another way. For instance, what if robots magically had access to the entire graph. Then would they still need honesty from other robots in order to settle down? It turns out that the answer is no on a ring, as demonstrated in~\cite{MMM20-byzantine,MMM21-TCS} in their algorithm \textsc{Time-Opt-Ring-Dispersion}. We generalize that algorithm to all graphs in Section~\ref{sec:dispersion-using-map}. In essence, if a robot $R$ can keep track of where all it has been, it can keep track of where other robots claimed to settle down. If $R$ sees a robot claiming to settle down somewhere where it should not, $R$ can easily discern that it is dealing with a Byzantine robot. Furthermore, we show that magic isn't always needed to allow a single robot to  learn the entire graph structure in Section~\ref{sec:find-map}, at least for a certain class of graphs. We borrow a technique from~\cite{CKP12} that allows a single robot to construct the quotient graph (defined in Section~\ref{sec:find-map}) of a given graph. For the class of graphs whose quotient graph is isomorphic to the original graph, this allows each robot to know the original graph structure.

However, to learn the entire graph for any graph, we require some communication. And so we show how to deal with dishonesty from Byzantine robots when constructing the graph in Section~\ref{sec:alg-general-graphs}. First we gather robots using the gathering technique in~\cite{DPP14}. We then use the robot and token paradigm from~\cite{DPP14} to construct a map of the graph. Two robots $R$ and $R'$ pair up and one acts as a robot that can explore the graph and the other acts as a movable token. They work together to construct a map and then return to the node where all robots were gathered. Then $R$ pairs up with another robot $R''$ and repeats the process. In this manner, $R$ pairs up with the other gathered robots and forms a corresponding number of maps of the graph. Among all these maps, the map that was constructed the majority of times is taken to be the map of the graph. And since there are at most  $\lfloor n/2 - 1 \rfloor$ Byzantine robots, each robot will correctly determine the map of the graph.

\begin{sloppypar}
The previous algorithm works for any arbitrary graph but takes $O(n^4 |\Lambda_{good}| X(n)) = \Tilde{O}(n^9)$ rounds to complete. What is the bottleneck that makes things so slow? It turns out that gathering slows us down dramatically. If robots are initially gathered together in the same node, the running time comes down to $O(n^4)$ rounds. If we further relax the number of Byzantine robots we can tolerate to $\lfloor n/3 - 1 \rfloor$, then it turns out that by tweaking the map finding algorithm, we can solve the problem in $O(n^3)$ rounds. Essentially, instead of having each robot pair up with each of the remaining robots (which may take $O(n)$ total pairings), we have the robots work together to form three groups which work together to form just $O(1)$ maps (which takes $O(1)$ pairings). And again, the map that is constructed the majority of times is taken as the map of the graph.
\end{sloppypar}

As mentioned, the bottleneck to the algorithm for general graphs is gathering. If we restrict the number of Byzantine robots to be $O(\sqrt{n})$, then we can use the gathering algorithm of~\cite{HTNOI20} coupled with our techniques to solve Byzantine dispersion in $O((f +  |\Lambda_{all}|) X(n)) = \Tilde{O}(n^5 \sqrt{n})$ rounds when robots are initially arbitrarily located on the graph.

Finally, in Section~\ref{sec:alg-strong-byz-robots}, we extend previous ideas to handle strong Byzantine robots. However, the trade-off is that our algorithms can only tolerate at most $\lfloor n/4 \rfloor -1$ Byzantine robots.

\subsection{Related Work}\label{subsec:rel-work}

The problem of dispersion of mobile robots on a graph was first introduced by Augustine and Moses Jr. \cite{Augustine:2018} and they provided solutions for various types of graphs both when robots were initially gathered at one node and when they started at arbitrary locations. Subsequently, the problem was studied by several papers \cite{AAMSS18,Kshemkalyani,KMS2019,KMS2020-ICDCN,KMS2020,tamc19,SSKM20,DBS21} in various settings to improve the efficiency of the solutions. The best known time-memory efficient algorithm for dispersion of $k\leq n$ robots on an arbitrary $n$-node graph when termination is required has time complexity of $O(\min\{m, k\Delta\}\log n)$ rounds and $O(\log n)$ bits of memory per robot~\cite{KMS2019}, where $\Delta$ is the maximum degree of the graph.
Recently, the problem of Byzantine dispersion was introduced by Molla et al.~\cite{MMM20-byzantine,MMM21-TCS} and time and memory optimal solutions were proposed for rings.

Dispersion has also been studied in the presence of crash faults~\cite{PSM21} as well as other models where either the graph was dynamic~\cite{AAMSS18}, the communication was global~\cite{KMS2020-ICDCN,KMS2020}, or randomization was used~\cite{tamc19,DBS21}.

Some problems that are closely related to dispersion on graphs are exploration, scattering, and gathering. The problem of graph exploration has been studied extensively in the literature for arbitrary graphs and usually assumes robots are initially gathered,
e.g., \cite{Bampas:2009,Cohen:2008,Das13,Dereniowski:2015,Fraigniaud:2005,MencPU17}. While some of these exploration algorithms (especially those which fit the current model) can be adapted to solve dispersion (with additional work), they, however, provide inefficient time-memory bounds; a detailed comparison is given in \cite{KMS2019}.
Another problem related to dispersion is scattering (also known as uniform-deployment) of mobile robots in a graph and is also studied by several papers, e.g., it has been studied for rings \cite{ElorB11,Shibata:2016} and grids \cite{Barriere2009,Poudel18} under different assumptions.

Gathering of robots is especially related, as it can be used as a subroutine when solving dispersion, as is seen in this paper. While it appears that the notion of Byzantine robots is not new to the mobile robots literature in general (e.g.,~\cite{BPT09,ABCTU13,CGKKNOS16}), it appears that its usage in the context of mobile robots on a graph is fairly recent.
To the best of our knowledge, apart from the problem of Byzantine dispersion that was recently introduced, only the problem of gathering has been studied in the context of Byzantine robots in the graph setting. Specifically, Dieudonn\'e et al.~\cite{DPP14} introduced the notion of Byzantine robots to the gathering problem. The paper mainly investigates the possibility and impossibility of gathering of mobile robots on graphs in the presence of Byzantine robots. They present some possibility results under certain assumptions on the minimum number of non-Byzantine robots present and on whether the value of $n$ is known to the robots. In particular, when $n$ is known, they develop an algorithm to gather strong Byzantine robots in an exponential number of rounds and another algorithm to gather weak Byzantine robots in $4n^4\textbf{P}(n,|\Lambda_{good}|)$ rounds.\footnote{Note that $\textbf{P}(n,|\Lambda_{good}|) = O(|\Lambda_{good}| X(n))$ (see~\cite{HTNOI20}).} In our setting, the latter algorithm takes $\Tilde{O}(n^9)$ rounds. There are some follow-up papers on Byzantine gathering, mostly focused on the feasibility of the solutions, e.g., \cite{BDD16,BDL18,DFLMS19}. Hirose et al.~\cite{HTNOI20} focus on reducing the run time of algorithms when the number of Byzantine robots is restricted. Some other interesting results study gathering with Byzantine robots in different models~\cite{DLM15,TOI18}.

\subsection{Paper Organization}\label{subsec:paper-org}
We show how to solve Byzantine dispersion in the presence of up to $n-1$ Byzantine robots for a certain class of graphs in Section~\ref{sec:alg-upto-n-1}. In Section~\ref{sec:alg-general-graphs}, we present several algorithms for general graphs. We show how to tolerate strong Byzantine robots in Section~\ref{sec:alg-strong-byz-robots}. We present our impossibility result in Section~\ref{sec:impossibility-result}. Finally, we present conclusions and future directions of work in Section~\ref{sec:conclusion}.

\section{Byzantine Dispersion: Tolerating any Number of Byzantine Robots}
\label{sec:alg-upto-n-1}
We devise an algorithm which solves dispersion while tolerating up to $n-1$ Byzantine robots. However, this algorithm works only for the class of graphs whose quotient graph is the same as (or isomorphic to) itself. The algorithm uses two procedures as building blocks. The first procedure, which is developed by Czyzowicz et al. \cite{CKP12}, is based on computing a map of the anonymous graph. In a nutshell, each robot computes a map of the graph independently and parallelly starting from its initial position. It takes time polynomial in $n$. Let us call this procedure {\sc Find-Map}. Then the second procedure disperses the robots on the graph (again independently and parallelly) with the help of the maps. In this procedure, each robot independently finds a position (node) on the graph to be settled in $O(n)$ time. Let us call this second procedure {\sc Dispersion-Using-Map}. We discuss them in detail below.

\medskip
\subsection{\textbf{Procedure} \textsc{Find-Map}}\label{sec:find-map}
It is difficult for a robot operating on a graph $G$ to compute a map (i.e., an isomorphic copy) of $G$. However, Czyzowicz et al. \cite{CKP12} present a protocol to construct the quotient graph of an anonymous graph $G$ by a single robot. Intuitively, this graph contains all the information that can be gained about $G$ from its exploration by a single robot. The quotient graph can be used in a manner similar to a map or sketch of the given graph $G$.

The following definition of quotient graph is adapted from \cite{YK96,CKP12}.  Let $G$ be a graph. The quotient graph of $G$, denoted $Q_G$, is a (not necessarily simple) graph defined as follows. Nodes of $Q_G$ correspond to sets of nodes of $G$ which have the same view. For any pair (not necessarily distinct) of nodes $X$ and $Y$ of $Q_G$, corresponding to some two subsets of nodes of graph $G$, there is an edge between $X$ and $Y$ with labels $p$ at $X$ and $q$ at $Y$, if there exists an edge $(x, y)$ in $G$ with $x \in X, y \in Y$ and with ports $p$ at $x$ and $q$ at $y$. Graphs $G$ and $H$ are called equivalent, if $Q_G = Q_H$, i.e., they are isomorphic.

Czyzowicz et al. \cite{CKP12} show that a single robot, with sufficient amount of memory, can construct the quotient graph of an anonymous graph with the help of the port labeling.
The following result follows from Section~5 of \cite{CKP12}.

\begin{lemma}\label{lem:quotient-graph}
Given an $n$-node anonymous graph with port labeling, there is a polynomial (in $n$) time algorithm to construct the quotient graph of $G$  by a single robot with $O(m\log n)$ bits of memory, where $m$ is the number of edges in $G$.
\end{lemma}

We can easily adapt this result to our Byzantine setting. Each robot independently and parallely runs the above protocol and constructs the quotient graph of $G$. However, for the dispersion problem, it is required that the quotient graph must be the same or an isomorphic copy of $G$ so that robots may not incorrectly settle at the same node (recall that a single node of the quotient graph could be a set of nodes of $G$).\footnote{It is also a requirement for the rendezvous and graph exploration problems \cite{CKP12,YK96}.} Thus, we get the following Corollary.

\begin{corollary}\label{cor:find-map}
Given a $n$-node anonymous graph, $n$ robots, placed arbitrarily on G, each with $O(m\log n)$ bits of memory, can construct a map (or the quotient graph) of $G$ independently and parallely in polynomial (in $n$) rounds.
\end{corollary}

Note that while having a map makes it easier to explore the graph \cite{YK96}, it won't help to solve the dispersion problem straightforwardly. The key reason is that the Byzantine robots can occupy a node before a non-Byzantine robot moves to that node. So, intuitively, a robot may not find an empty node, but be forced to settle at an occupied node. Thus, a robot needs to distinguish between non-Byzantine and Byzantine robots, which is a challenging task since a Byzantine robot can behave maliciously by sending incorrect information. We now present a protocol which overcomes these challenges.

\medskip
\subsection{\textbf{Procedure} {\sc Dispersion-Using-Map}}\label{sec:dispersion-using-map}
Assume that each robot has a map of the graph before starting this procedure and the robots start this procedure at the same time (since it is a synchronous system). The maps may be distinct from each other and the robots may be positioned at different nodes. A robot locally computes a spanning tree (say, a DFS tree) on the map of the graph in its memory, where the root of the spanning tree is the current position (node) of the robot.  Each robot locally marks the nodes of the graph (map) by assigning some identities. Without loss of generality, assume that the node IDs, maintained by a robot, are $\{v_1, v_2, \dots, v_n\}$. The idea is that each robot traverses its own spanning tree to find an empty node to settle down. Assume the robots perform DFS tree traversal in parallel. When multiple robots move to an empty node, the minimum ID robot has the preference to settle down at that node. However, this simple idea will not work is in the presence of Byzantine robots and needs to be carefully adapted.

During the DFS tree traversal, a robot stores the IDs of the settled robots and the corresponding nodes. For this, each robot $r$ maintains an array $A_r$ of size $n$, where the entry $A_r[v_i]$ stores the IDs of the settled robots at the node $v_i$. Note that a node may have multiple settled robots, since the Byzantine robots may settle with a non-Byzantine robot (to interrupt the protocol). Initially, the array is empty. Recall that we assume the presence of {\em weak} Byzantine robots, i.e., the robots can act in a Byzantine manner by behaving maliciously to thwart the algorithm but they are weak in that they cannot fake their IDs. Thus, if a (non-Byzantine) robot encounters a settled robot, say, $\hat{r}$ at some different node, then it can surely blacklist $\hat{r}$ as a Byzantine robot. For this, each robot $r$ also maintains a set $B_r$ which contains the blacklisted Byzantine robots by $r$.\footnote{A robot may also detect a Byzantine robot by some other means, say, when a robot does not follow the protocol, e.g., a robot that does not transmit a message when it is supposed to do so.} Furthermore, each robot maintains a state which takes one of two values: (i) {\em Settled}, and (ii) {\em tobeSettled}. Each robot also maintains a binary $0-1$ flag. The initial state of a robot is tobeSettled and the flag is set to $0$. A robot may change its state to Settled at some node depending on the other co-located robots' IDs and the states. In the same way, whenever a robot moves to a node $v$, it may change its flag value to $1$, which indicates that the robot {\em intends} to settle at $v$. A settled robot never changes its position and the state, which is Settled.

Let us consider a robot $r$ and a node $v$ where $r$ enters $v$ in some round $t$ during its DFS tree traversal. We describe how robot $r$ determines whether it settles down at $v$ or continues the traversal. Let $S^s(v)$ be the set of all robots with state Settled, and $S^{tbs}(v)$ be the set of all robots with state tobeSettled at node $v$ in round $t$. We simply denote them by $S^s$ and $S^{tbs}$ when it is clear from the context. For simplicity of our algorithm, we assume that each round is broken down into $n$ sub-rounds. In the following procedure, we mention robots \textit{waiting} at a node and then making a decision. Formally, for the robots present at a node in a given round, there is a total order on their IDs. Each robot of rank $\mathcal{Y}$ waits until sub-round $\mathcal{Y}$ before it makes its decision. The procedure for $r$ at node $v$ (until $r$ settles down) is given below:

\begin{enumerate}
    \item If $S^s = \emptyset$ and, either $r$ is the minimum in $S^{tbs}$ or $S^{tbs} \setminus \{r\} \subseteq B_r$, then $r$ settles at $v$ changes its state to $Settled$.

    \item Else, if $S^s = \emptyset$ and, neither $r$ is minimum in $S^{tbs}$, nor $S^{tbs} \setminus \{r\} \subseteq B_r$, then:

        \subitem $\rhd$ a) If all the robots in $S^{tbs} \setminus \{r\}$, whose ID is smaller than $r$, are in $B_r$, then $r$ settles at $v$ and changes its state to $Settled$.

        \subitem $\rhd$ b) Else, $r$ sets its flag value to $1$ and checks if any other robots in $S^{tbs} \setminus B_r$ has flag value $1$. If NO, then $r$ settles at $v$ and changes its state to $Settled$. If YES,
        then $r$ waits
        and observes the smaller ID robots in $S^{tbs} \setminus B_r$ (other than $r$).
        If any of them changes its state to $Settled$, then $r$ stores the ID(s) of the settled robot(s) in $A_r[v]$ and moves to the next node following the DFS tree traversal.
        If no smaller ID robot changes its state to $Settled$,
        then $r$ settles at $v$ and changes its state to $Settled$.

        \item If $S^s \neq \emptyset$, then:

        \subitem $\rhd$ a) If $S^s \subseteq B_r$ and $r$ is the minimum in $S^{tbs} \setminus B_r$ then $r$ settles at $v$ and changes its state to $Settled$.

        \subitem  $\rhd$ b) Else, if $S^s \subseteq B_r$ and $r$ is NOT the minimum in $S^{tbs} \setminus B_r$, then: $r$ sets its flag value to $1$ and checks if any robot in $S^{tbs} \setminus B_r$ has flag value $1$. If NO, then $r$ settles at $v$ and changes its state as $Settled$. If YES,
        then $r$ waits
        and observes the smaller ID robots in $S^{tbs} \setminus B_r$.
        If any of them changes its state to $Settled$ then $r$ stores the settled robot(s) in $A_r[v]$ and moves to the next node following the DFS traversal. If no smaller ID robot changes its state to $Settled$,
        then $r$ settles at $v$ and changes its state to $Settled$.

        \subitem  $\rhd$ c) Else, if $S^s \nsubseteq B_r$, then $r$ stores the IDs of the settled robots from $S^s \setminus B_r$ in $A_r[v]$ and moves to the next node.

    \item If $r$ finds any robots $\hat{r} \in \cup_{\{v_i : v_i \neq v\}} A_r[v_i]$ at $v$ (i.e., settled earlier at some node before $v$), then $r$ blacklists $\hat{r}$ and stores it in $B_r$. Furthermore, $r$ also considers a robot $\hat{r}$ as Byzantine if $\hat{r}$ does not  transmit a message when it is supposed to do so.
\end{enumerate}
The above procedure is executed by all robots in parallel, when traversing their respective trees. We analyze the correctness of the procedure {\sc Dispersion-Using-Map} below. Let us first make the following Observation.

\begin{observation}\label{obs:initial-position}
If a single robot is positioned at a node, it settles at that node.
\end{observation}
The above Observation follows from Step~1 of the procedure.

\begin{lemma}\label{lemma:different-blacklist}
A non-Byzantine robot can never be in the blacklisted set of any other non-Byzantine robots.
\end{lemma}
\begin{sloppypar}
\begin{proof}
Let $r$ be  a non-Byzantine robot. For a non-Byzantine robot $r_i$, the blacklisted set $B_{r_i}$ contains the Byzantine robots encountered by $r_i$. We claim that $r \notin B_{r_i}$ for any non-Byzantine robot $r_i$. From Observation~\ref{obs:initial-position} and the fact that ``a settled robot never changes its position and state'', $r$ never changes its position and state once it is settled. Furthermore, a non-Byzantine robot never deviates from the algorithm. Thus, from Step~4 of the procedure, $r \notin B_{r_i}$ for any non-Byzantine robot $r_i$.
\end{proof}
\end{sloppypar}
The correctness of the procedure follows from the following lemma.

\begin{lemma}\label{lem:correctness-upto-n-1}
No two non-Byzantine robots settle at the same node $v$.
\end{lemma}
\begin{proof}
Let two non-Byzantine robots $r_1$ and $r_2$ be co-located at some node $v$. Then by Lemma~\ref{lemma:different-blacklist}, $r_1$ is not in the blacklist of $r_2$ and vice-versa. If one
of them is already settled at $v$, say $r_1$, then $r_2$ will
not settle there, because of the check in Step~3.c.
If neither of them is settled at $v$, then it is impossible for both of them to be settled at $v$ at the same time. Suppose the ID of $r_1$ smaller than $r_2$. Then by Step~2.b, $r_2$ does not settle at $v$.
\end{proof}

\begin{lemma}\label{lem:time-upto-n-1}
The procedure {\sc Dispersion-Using-Map} correctly disperses the robots in $O(n)$ time.
\end{lemma}
\begin{proof}
The procedure solves Byzantine dispersion since at any node, at most one non-Byzantine robot settles (by Lemma~\ref{lem:correctness-upto-n-1}).

Since the normal DFS tree traversal takes at most $2n-1$ steps to visit all the nodes of the graph (when the map is known), it takes $O(n)$ rounds to visit all the nodes of the graph. A robot will find a node to settle down after visiting all the nodes of the graph. The reason is that whenever the robot visits a node $v$, it stores at least one settled robot's ID in its array $A[v]$ (if the robot itself does not settle at $v$). Moreover, a non-Byzantine robot never changes its settled position (node) and if a (settled) robot changes its position, it may be blacklisted. Hence, by the pigeonhole principle, the robot must find a node to settle down at after visiting all the $n$ nodes of the graph. Therefore, the procedure {\sc Dispersion-Using-Map} correctly disperses the robots and takes at most $O(n)$ rounds.
\end{proof}

\begin{theorem}\label{thm:main-upto-n-1}
There is an algorithm to solve Byzantine dispersion of $n$ robots, when up to $n-1$ of them are weak Byzantine robots, on an $n$ node graph in the class of graphs where the graph is isomorphic to its quotient graph. Moreover, the time complexity of the algorithm is polynomial in $n$.
\end{theorem}
\begin{proof}
\begin{sloppypar}
Each robot performs the two procedures {\sc Find-Map} and {\sc Dispersion-Using-Map}. Since each robots performs {\sc Find-Map} independently without relying on any information from the other robots or their movements, no amount of Byzantine robots can prevent a non-Byzantine robot from generating a quotient graph. Once the graph is created, by Lemma~\ref{lem:time-upto-n-1} we see that {\sc Dispersion-Using-Map} allows the non-Byzantine robots to solve Byzantine dispersion. Thus, the algorithm works in the presence of any number of Byzantine robots, i.e., $f \leq n-1$.
\end{sloppypar}

The time complexity follows from Corollary~\ref{cor:find-map}, which shows that{\sc Find-Map} takes polynomial in $n$ rounds, and Lemma~\ref{lem:time-upto-n-1}, which shows that {\sc Dispersion-Using-Map} takes $O(n)$ rounds. In fact, the polynomial term dominates the $O(n)$ bound.
\end{proof}

\begin{remark}[Memory per Robot]\label{rem:memory-upto-n-a}
\begin{sloppypar}
A robot stores the map of the graph, which requires $O(m\log n)$ bits, where $m$ is the number of the edges of the graph. It also stores the IDs of all the robot, which requires $O(n\log n)$ bits, assuming the IDs are at most polynomially large.  Other parameters (like state, synchronous round counter, etc.) take at most $O(\log n)$ bits. Thus, a robot requires $O(m\log n)$ bits of memory.
\end{sloppypar}
\end{remark}

\section{Byzantine Dispersion on General Graphs}
\label{sec:alg-general-graphs}
In this section we discuss solutions to the Byzantine dispersion problem on any port labeled graph $G$ with $n$ nodes. First we propose an algorithm that handles up to $\lfloor {n/2}-1 \rfloor$ Byzantine robots when robots are initially placed arbitrarily. The solution takes $O(n^4 |\Lambda_{good}| X(n))$ rounds. We also see that if the robots are initially gathered at some node, then the round complexity becomes $O(n^4)$ (see Section \ref{sec:n/2Byz}). Later, we show that if the number of Byzantine robots is bounded to be at most $\lfloor {n/3-1} \rfloor$, then Byzantine dispersion can be solved much faster in $O(n^3)$ rounds, provided that all the robots start from a gathered configuration (see Section~\ref{sec:n/3Byz}).\footnote{Note that our algorithms still work even if only the non-Byzantine robots are gathered in one node.} Finally, in Section~\ref{sec:sqrt(n)Byz}, we show that it is possible to solve Byzantine dispersion in $O((f +  |\Lambda_{all}|) X(n))$ rounds when robots start from an arbitrary initial configuration if we limit the number of Byzantine robots to be at most $O(\sqrt n)$.

All the algorithms have a similar outline. In Phase~1, we first gather the non-Byzantine robots using existing algorithms \cite{DPP14,HTNOI20} in case the robots do not start in a gathered configuration. Then in Phase~2, robots construct a map of the graph using the techniques that we discuss in this section. In fact, we have different map finding techniques based on different upper bounds on the number of Byzantine robots.
Finally, robots disperse with the help of the map using our procedure {\sc Dispersion-Using-Map} (refer Section \ref{sec:dispersion-using-map}) in Phase~3.

\subsection{$\lfloor {n/2}-1 \rfloor$ Byzantine robots}\label{sec:n/2Byz}
In this section we discuss gathering of $\lfloor {n/2}-1 \rfloor$ Byzantine robots starting from any arbitrary initial configuration.\\
\\
\noindent\textbf{Phase 1 (gathering):} All the non-Byzantine robots are gathered in some node using the existing algorithm in \cite{DPP14} that handles any number of weak Byzantine robots on any graph of known size $n$. Their algorithm requires $4n^4\textbf{P}(n,|\Lambda_{good}|)$ rounds to terminate for the non-Byzantine robots where $|\Lambda_{good}|$ is the length of the largest ID of any non-Byzantine robot. This is equivalent to $O(n^4|\Lambda_{good}| X(n))$ (see~\cite{HTNOI20}), where $X(n)$ is the number of rounds needed to explore any graph of size $n$.
In our case the largest ID of a non-Byzantine robot can be as large as $n^c$ for some constant $c>1$ (i.e., length is $O(\log n)$) and hence all non-Byzantine robots gather and terminate in $O(n^4|\Lambda_{good}| X(n))$ rounds in this phase. The exact number of rounds before all non-Byzantine robots terminate in \cite{DPP14} depends on the size of the graph $n$. Let this gathering take $T_1$ rounds. All non-Byzantine robots wait for $T_1$ rounds for the gathering algorithm to be complete.\\
\\
\noindent\textbf{Phase 2 (map finding):} Phase two starts after $T_1$ rounds have been completed from the beginning of the algorithm. In this stage, we have robots pair up in order to find a map. Suppose two robots work together to run the exploration algorithm with movable token from \cite{DPP14}. One robot acts as an agent and the other acts as a movable token. It is possible for the pair to compute an isomorphic map of the underlying graph in $O(n^3)$ rounds if both $i$ and $j$ are non-Byzantine. Let $T_2$ ($ = O(n^3)$) be an upper bound on the time taken for two non-Byzantine robots starting at a given node to perform the algorithm from~\cite{DPP14} and subsequently return to the starting node.\footnote{All non-Byzantine robots remember the port numbers they traveled through during this map finding subroutine. After constructing the map, the robots may use those port numbers to get back to the node where they were gathered at the beginning of this phase. This may take at most another $O(n^3)$ rounds.} Each robot $R$ pairs up with every other robot to form a map.
$R$ then takes the map formed by the majority of pairings and uses it as the map of the graph.

We now develop an algorithm to have every robot pair up with every other robot (and perform the map finding algorithm) in $O(n^4)$ rounds.
The algorithm proceeds in $\lceil \log n \rceil$ stages. Let all the robots initially form a large group. In the first stage, we divide the group into two subgroups of equal size (if they are unequal, we add a dummy robot to the smaller subgroup to make them equal) and have them form pairs with each other until every robot in each subgroup has paired with every robot in the other subgroup. In the next stage, each of these subgroups is further divided into equal size subgroups and the process is repeated. Eventually, all subgroups will contain one robot and no further division is possible. Formally, consider a group of robots $G$ (possibly one of several) present at the beginning of stage $i$. The robots in group $G$ are divided into groups $G^0$ and $G^1$ of sizes $\lceil G/2 \rceil$ and $\lfloor G/2 \rfloor$ respectively. If $|G^0| \neq |G^1|$, add a dummy value to $G^1$ so that they are equal in size. Let the robots of $G^0$ be $G^0_1, G^0_2, \ldots, G^0_{\lceil G/2 \rceil}$ and those of $G^1$ be $G^1_1, G^1_2, \ldots, G^1_{\lceil G/2 \rceil}$. Now, do the following for the next $(n/2^i + 2)T_2$ rounds: from round $j T_2 + 1$ to round $(j+1)T_2$, have each robot $G^0_x$ pair up with robot $G^1_{x+j}$ and run the algorithm from~\cite{DPP14} and return to the start node.\footnote{If one of the robots in a pairing is Byzantine, it may not correctly follow the algorithm. To prevent this affecting non-Byzantine robots, each robot keeps track of how many rounds are needed to run~\cite{DPP14}. After this many rounds, the robot automatically stops running the algorithm and returns to the start node. Also, if all possible pairings are done before the end of the stage, the robot waits at the start node until the next stage begins.}

Let us say that phase 2 terminates after $T_3$ rounds from the beginning of the algorithm.\\

\noindent\textbf{Phase 3 (dispersion):} Phase 3 starts after $T_3$ rounds are over from the beginning of the algorithm. In this phase, robots disperse using our procedure {\sc Dispersion-Using-Map} (refer to Lemma~\ref{lem:time-upto-n-1}) in $O(n)$ rounds using the isomorphic maps computed in phase 2.\\
\\
\noindent\textbf{Correctness:} The correctness of phase 1 follows from~\cite{DPP14} and the correctness of phase 3, assuming robots have an isomorphic map of the graph, follows from Lemma~\ref{lem:time-upto-n-1}. Thus, we only need to prove that in phase 2, each robot computes a correct map.

We show that the algorithm for pairing robots pairs every robot with every other robot and the total time taken is $O(n^4)$ rounds. In each stage, it is easy to see that for a given group that is divided into two subgroups, each robot of each group pairs up with each robot of the other group. In $\lceil \log n \rceil$ stages, each group will reach the size of $1$, thus each robot will have paired with every other robot. As for time, stage $i$ takes $n/2^i + 2)T_2$ rounds. So all $\lceil \log n \rceil$ stages take $O(n + \log n)T_2 = O(n^4)$ rounds.

Of these pairs of robots that perform map finding, call a pair a \textit{good pair} if both the robots are non-Byzantine, else call it a \textit{bad pair}. As the number of Byzantine robots are at most $\lfloor {n/2}-1 \rfloor$, the number $(n-f)$ of possible good pairs involving any robot is at least one more than the number of bad pairs. In turn each non-Byzantine robot computes at least $(n-f)$ maps which are isomorphic to the actual graph $G$. As all these $(n-f)$ maps, $f \le \lfloor {n/2}-1 \rfloor$,  are isomorphic to each other, each robot can find the correct map isomorphic to $G$ by using a majority rule over all the  maps (at most $n-1$) it generated in phase 2 of the algorithm.

As the overall round complexity is dominated by phase 1 of the algorithm, we have the following theorem.
\begin{theorem}\label{theorem:n/2}
\begin{sloppypar}
Byzantine dispersion can be solved in $O(n^4 |\Lambda_{good}| X(n))$ synchronous rounds in the presence of at most $\lfloor {n/2}-1 \rfloor$ weak Byzantine robots when robots start from any arbitrary initial configuration.
\end{sloppypar}
\end{theorem}

\noindent \textbf{Initially gathered:} Note that, if robots starts from a gathering configuration, then they only run phase 2 and phase 3 of Section \ref{sec:n/2Byz} to disperse. As phase 1 is not required, the round complexity becomes $O(n^4)$. We have the following theorem.
\begin{theorem}\label{theorem:n/21}
Byzantine dispersion can be solved in $O(n^4)$ synchronous rounds in the presence of at most $\lfloor {n/2}-1 \rfloor$ weak Byzantine robots when all robots are initially in a gathered configuration.
\end{theorem}

\subsection{$\lfloor {n/3-1} \rfloor$ Byzantine robots: A faster algorithm}\label{sec:n/3Byz}
If initially the non-Byzantine robots are gathered and there are at most $\lfloor {n/3-1} \rfloor$ Byzantine robots, then we can generate a faster algorithm to solve dispersion. By modifying the map finding subroutine as described below, we have an algorithm that solves Byzantine dispersion in $O(n^3)$ rounds.\\
\\
\noindent\textbf{Phase 1 (map finding):}
Let $k$ robots (including at least $(n-f)$ non-Byzantine robots) be gathered at some node $s$, so $k\ge n-f$.
The robots make 3 groups of at least $\lfloor k/3 \rfloor$ robots each, namely $A$, $B$ and $C$.
Without loss of generality, let the $\lfloor k/3 \rfloor$ smallest ID robots be in group $A$. From the remaining robots, let the smallest $\lfloor k/3 \rfloor$ ID robots form group $B$ and the remaining robots form group $C$.
Each robot remembers which robot belongs to which group among those $k$ gathered robots. As each robot remembers the IDs of the other $k-1$ initially gathered robots, if it interacts with a robot with an ID from outside this set of IDs, it can identify that robot as Byzantine.

Again the map finding algorithm with movable token (see \cite{DPP14}) is used. In each run of the map finding subroutine, one group of robots acts as an agent and another two groups of robots act together as a token. The robots that are acting as a token {\it believe} a group of at least $\lfloor k/6+1 \rfloor$ robots (i.e., more than half of the robots) from the group of robots that are working as the agent. Here, believe means that the robots acting as the token consider the instructions from such a group  of robots as the instruction from the agent and work as per those instructions.  On the other hand, the agent believes only a group of at least $\lfloor k/3+1 \rfloor$ robots (i.e., more than $\lfloor k/3 - 1 \rfloor$ many robots) from the group of robots that are working as a token.

\begin{sloppypar}
First $A$ works as an agent and $B\cup C$ works as a token (this group is formed in one round). The agent does the map finding with the help of the token. This takes $T_2$ rounds as in phase 2 of Section \ref{sec:n/2Byz}.
All non-Byzantine robots remember the port numbers they traveled through during this map finding subroutine. Using these port numbers, they return to the node where they were gathered at the beginning of this phase. This may take another $T_2$ rounds. Then the group $A\cup C$ is formed in one round. So from the $(2 T_2 + 2)$-th round onwards,  $A\cup C$ works as token and $B$ acts as an agent. After another $2T_2$ rounds, finally $B\cup A$ works as a token (takes one round to form this group) and $C$ acts as an agent. From round $4 T_2 + 4$ to round $6 T_2+3$, $C$ does the map finding with the help of $B \cup A$.
\end{sloppypar}

So each non-Byzantine robot computes at most three maps of the graph among which at least two are isomorphic to each other. Non-Byzantine robots choose one from those isomorphic maps. Since the map finding subroutine runs only thrice, the round complexity of this phase is at most $6T_2 +3$ which is of $O(n^3)$.  Let this phase terminate after $T'_3$ rounds from beginning of the algorithm.\\
\\
\noindent\textbf{Phase 2 (dispersion):} This phase starts in the $(T'_3+1)$-th round. In this phase, robots disperse using the procedure {\sc Dispersion-Using-Map} (refer to Lemma~\ref{lem:time-upto-n-1}) in $O(n)$ rounds using the isomorphic map of $G$ computed in phase 1.\\
\\
\noindent\textbf{Correctness:} We need to only prove the correctness of phase 1 as the correctness of phase 2, assuming each robot has access to an isomorphic map of the graph, was proved in Lemma~\ref{lem:time-upto-n-1}. Each robot remembers the set of other $k-1$ IDs of the collocated robots and also remembers which robot belongs to which group among those $k$ gathered robots.
If it interacts with a robot with an ID from outside the set of these $k-1$ IDs (i.e., an ID of one of the $n-k$ robots not present), it can identify that robot as Byzantine. Hence do not have any impact on the algorithm hereafter.

Observe that there cannot be more than $\lfloor k/3-1 \rfloor$ Byzantine robots among those $k$ robots. Hence only one group can have more than or equal to $\lfloor k/6+1 \rfloor$ many Byzantine robots. Because if two groups have $\lfloor k/6+1 \rfloor$ Byzantine robots each, then this contradicts the fact that the total number of Byzantine robots are less than or equal to  $\lfloor k/3 - 1 \rfloor$ among those $k$ robots.

Let $A$ have $\lfloor k/6+1 \rfloor$ or more Byzantine robots in it. Those Byzantine robots can only affect the map finding algorithm when $A$ works as an agent since the majority of the robots in $A$ are Byzantine. In the other two cases, both the group acting as an agent and the combined group acting as a token have a majority of non-Byzantine robots. As an example, when $B$ works as an agent and $A\cup C$ works as a token, then the majority of robots acting as an agent as well as those acting as a token are non-Byzantine. Whenever more than equal to $\lfloor k/6+1 \rfloor$ robots of group $B$ instruct the robots in $A\cup C$, the token will move. Note that, if some Byzantine robots change groups and join the Byzantine robots in $B$, making the number of Byzantine robots more than $\lfloor k/6+1 \rfloor$, then the non-Byzantine robots in $A\cup C$ understand it as each robot knows who is in which group. So, non-Byzantine robots of $A\cup C$ believe a group of $\lfloor k/6+1 \rfloor$ or more robots from group $B$ only.  Also, in  $A\cup C$ there are more than $\lfloor k/3+1 \rfloor$ non-Byzantine robots and the non-Byzantine robots of $B$ believe only such a group as token that is consisting of more than equal to $\lfloor k/3+1 \rfloor$ robots from $A\cup C$. So, the non-Byzantine robots of $B$ correctly find an isomorphic map of the graph and pass this information to other robots such that robots in $A\cup C$ also have this map. So, among three runs of the map finding subroutine, for sure in two instances, the non-Byzantine robots find a correct isomorphic map of the graph. Hence all the non-Byzantine robots would be able to choose a correct map by applying a majority rule to the three maps they have computed.

As the overall round complexity is dominated by phase 1 of the algorithm, we have the following theorem.
\begin{theorem}\label{theorem:n/3}
Byzantine dispersion can be solved in $O(n^3)$ synchronous rounds in the presence of at most $\lfloor {n/3}-1 \rfloor$ weak Byzantine robots when all robots start in a gathered configuration.
\end{theorem}

\subsection{$O(\sqrt n)$ Byzantine robots}\label{sec:sqrt(n)Byz}
When the number of Byzantine robots is known to be at most $O(\sqrt n)$, it is possible to solve Byzantine dispersion in $O((f +  |\Lambda_{all}|) X(n))$ rounds even if robots start from an arbitrary initial configuration. This is due to the gathering algorithm of \cite{HTNOI20}.
\\

\noindent\textbf{Phase 1 (gathering):} Use the gathering algorithm of \cite{HTNOI20}, which takes $O((f +  |\Lambda_{all}|) X(n))$ time to gather. It requires that all non-Byzantine robots start the algorithm at the same time, which is the case here.\\
\\
\noindent\textbf{Phase 2 (map finding):} This is similar to phase 1 (map finding) of Section \ref{sec:n/3Byz}.  After phase 1 (gathering), some $n - f \leq k \leq n$ robots gathered at one node including all the non-Byzantine robots. In this situation, we can create two groups of $\lfloor k/2 \rfloor$ and  $\lceil k/2 \rceil$ robots. One group works as agent and another group works as movable token. As there are enough non-Byzantine robots in both the groups, each non-Byzantine robot finds a correct map after one run of the map finding algorithm using exploration with movable token from \cite{DPP14}. Thus this phase takes $O(n^3)$ rounds.\\

\noindent\textbf{Phase 3 (dispersion):} In this phase, robots disperse using the procedure {\sc Dispersion-Using-Map} (refer to Lemma~\ref{lem:time-upto-n-1}) in $O(n)$ rounds using the isomorphic map of the graph computed in phase 2.\\
\\
\noindent\textbf{Correctness:} No correctness proof is required in this case as we already have correctness proofs of phase 1 from~\cite{HTNOI20}, phase 1 from Section~\ref{sec:n/3Byz}, and
the procedure {\sc Dispersion-Using-Map} from Lemma~\ref{lem:time-upto-n-1}.

As the overall round complexity is dominated by phase 1 of the algorithm, we have the following theorem.
\begin{theorem}\label{theorem:sqrtn}
\begin{sloppypar}
Byzantine dispersion can be solved in $O((f +  |\Lambda_{all}|) X(n))$ synchronous rounds in the presence of at most $O(\sqrt n)$ weak Byzantine robots when the robots start from any arbitrary initial configuration.
\end{sloppypar}
\end{theorem}

\section{Handling Strong Byzantine Robots}
\label{sec:alg-strong-byz-robots}
In this section we show that it is possible to solve Byzantine dispersion in the presence of at most $\lfloor{n/4 -1}\rfloor$ strong Byzantine robots in polynomial time if all robots are initially in a gathered configuration. We describe our algorithm below. Note that although we assume all $n$ robots are initially gathered, we describe the algorithm as if only a subset of the $n$ robots, including all non-Byzantine robots, are initially gathered. This is so that our algorithm can be paired with another algorithm later on which allows us to solve Byzantine dispersion from any arbitrary initial configuration.\\
\\
\noindent\textbf{Phase 1 (map finding):}
Let $k$ robots, including all non-Byzantine robots, be gathered at some node. Since there are at least $(n-f)$ non-Byzantine robots, $k\ge n-f$. Each robot remembers the IDs of the remaining $k-1$ gathered robots. If it interacts with a robot with an ID from outside the set of these $k-1$ IDs (i.e., an ID of one of the $n-k$ robots not present), it can identify that robot as Byzantine. The robots make an ordering of those $k$ IDs in increasing order such that the smallest ID robot is assigned to 1 and the largest ID robot assigned to $k$ in this order.

The robots make 2 groups, $A$ and $B$, of at least $\lfloor k/2 \rfloor$ robots each. The $\lfloor k/2 \rfloor$ smallest ID robots are in group $A$ and the remaining robots make group $B$. Each robot also remembers which robot belongs to which group.

Again, the map finding algorithm with movable token from \cite{DPP14} is used. Group $A$ works as the agent and group $B$ works as the movable token. Group $A$ believes robots from group $B$ as representing the position of the token only if there are at least $\lfloor {n/4} \rfloor$ robots
from group $B$ present. Similarly, the robots of group $B$ move as the token only when they get an instruction from a set of at least $\lfloor {n/4} \rfloor$ robots from group $A$. Map finding takes $O(n^3)$ rounds to terminate. After the map finding is done, all robots come back to the node from where they started (since they remember the port numbers they visited during map finding). Let this phase take $T$ rounds from the beginning of the algorithm (where $T = O(n^3)$).\\
\\
\noindent\textbf{Phase 2 (dispersion):} This phase starts in the  $(T+1)$-th round. In this phase, robots disperse using the following algorithm. It is similar to the rooted ring dispersion used in \cite{MMM20-byzantine,MMM21-TCS}. All the robots have a unique isomorphic map of the graph. The robots make a deterministic ordering of the nodes of the graph as $v(1), v(2), \cdots, v(n)$. The robot which was assigned to $i$ (in the deterministic increasing ordering of those initially gathered $k$ robots) decides to settle at node $v(i)$. Clearly this takes no more than $n$ rounds.\\
\\
\noindent\textbf{Correctness:} First we show that robots accurately generate a map in phase~1. Each robot remembers which robot belongs to which group among those $k$ gathered robots. As each robot remembers the IDs of the remaining $k-1$ gathered robots, if it interacts with a robot with an ID from outside this set of IDs, it can identify that robot as Byzantine.
Observe that each group have at least $\lfloor n/4 \rfloor$ non-Byzantine robots as all the non-Byzantine robots are divided in two groups. So both the groups have a majority of non-Byzantine robots. Even if robots change groups, both groups would still have a majority of non-Byzantine robots since the total number of Byzantine robots we allow is no more than $\lfloor{n/4 -1}\rfloor$.

As the robots are strong Byzantine, they can change their IDs and can even take IDs of non-Byzantine robots. But this does not hamper any communication as both group of robots trust only a group of at least $\lfloor n/4 \rfloor$ robots from the other group. Even if Byzantine robots duplicate IDs, still as a group they can not make it equal to  $\lfloor n/4 \rfloor$. At the end of phase 1, all the non-Byzantine robots are back in their initial positions, i.e., all are gathered at a single node again.

Now we look at the correctness of phase 2, i.e., we show that robots achieve Byzantine dispersion. The IDs of the robots are in $[1, n^c]$ and the non-Byzantine robots have distinct IDs. The ordering that initially gathered $k$ robots do on their IDs in increasing order, is always unique.
Also all the non-Byzantine robots have the unique isomorphic map (hence all compute exactly the same ordering of the nodes), the dispersion is correct.

The first phase takes $O(n^3)$ rounds and the second phase takes $O(n)$ rounds. Totally, the round complexity of the algorithm is $O(n^3)$ rounds.

Thus we have the following theorem.
\begin{theorem}\label{theorem:n/4s}
Byzantine dispersion can be solved in $O(n^3)$ synchronous rounds in the presence of at most $\lfloor{n/4 -1}\rfloor$ strong Byzantine robots when all robots start in a gathered configuration.
\end{theorem}

\noindent \textbf{Arbitrary initial configuration:} We report a method that handles at most $\lfloor{n/4 -1}\rfloor$ strong Byzantine robots and solves the Byzantine dispersion problem, though it takes exponential rounds. We first gather robots using the algorithm from \cite{DPP14} that gathers at most $\lfloor{n/3 -1}\rfloor$ strong Byzantine robots in any graph of known size starting from an arbitrary initial configuration. Robots create groups such that each of them contains at least $f+1$ non-Byzantine robots and then gathering of those groups ensures gathering of all non-Byzantine robots. However the round complexity is exponential and the knowledge of $f$ is required in this case. Once the gathering is done, we may use the algorithm we described earlier to solve Byzantine dispersion when robots are initially gathered. As the correctness of both the map finding and dispersion phase with strong Byzantine robots are given above and the correctness of the gathering algorithm comes from~\cite{DPP14}, we have the following theorem.

\begin{theorem}\label{theorem:n/4s1}
Byzantine dispersion can be solved deterministically in an exponential number of rounds in the presence of at most $\lfloor{n/4 -1}\rfloor$ strong Byzantine robots when robots start in any arbitrary initial configuration and $f$ is known to all the robots.
\end{theorem}

\section{Impossibility Result}
\label{sec:impossibility-result}
Suppose that instead of restricting the number of robots to $n$, we now have $k$ robots that are trying to solve Byzantine dispersion on an $n$ node graph, where up to $f$ of the robots are Byzantine.\footnote{Note that $k$ can be either $\leq n$ or $> n$. However, the result is only truly meaningful when $k > n$.} The problem of Byzantine dispersion is slightly modified such that in the final configuration, each node should have at most $\lceil (k-f)/n \rceil$ non-Byzantine robots on it.  We show that, even if the robots know the values of $n$, $k$, and $f$, and the Byzantine robots are weak, if there are too many Byzantine robots, then deterministically solving Byzantine dispersion is impossible.

\begin{theorem}
There does not exist any deterministic algorithm to solve Byzantine dispersion when $\lceil k/n \rceil > \lceil (k-f)/n \rceil$.
\end{theorem}

\begin{proof} Let there exist some algorithm $A$ that solves Byzantine dispersion of $k$ robots on any $n$ node graph even when $\lceil k/n \rceil > \lceil (k-f)/n \rceil$. Let us consider one execution of this algorithm on $k$ robots when $f=0$. Now, consider one node where $\lceil k/n \rceil$ robots with IDs $R_1, R_2, \ldots, R_{\lceil k/n \rceil}$ settle. Consider another execution of this algorithm where robots $R_1, R_2, \ldots, R_{\lceil k/n \rceil}$ are non-Byzantine and some $f$ robots are Byzantine such that $\lceil k/n \rceil > \lceil (k-f)/n \rceil$. Let the Byzantine robots act like the non-Byzantine robots from the previous execution of $A$, resulting in the robots $R_1, R_2, \ldots, R_{\lceil k/n \rceil}$ settling on the same node. Now, $\lceil k/n \rceil$ non-Byzantine robots settled on the same node, but Byzantine dispersion requires at most $\lceil (k-f)/n \rceil$ non-Byzantine robots to settle on the same node and here $\lceil k/n \rceil > \lceil (k-f)/n \rceil$, resulting in a contradiction.
\end{proof}

\section{Conclusion and Future Work}
\label{sec:conclusion}
In this paper, we have demonstrated how Byzantine dispersion may be solved deterministically in arbitrary graphs. Depending on whether the Byzantine robots are weak or strong, whether the robots are initially gathered or not, and the upper bound on the number of Byzantine robots present, we have developed several algorithms that exploit those conditions to solve the problem as quickly as possible.

In terms of future directions, there are three broad avenues to look into. First, it appears that gathering is a major bottleneck with respect to time when we want to solve the problem. To solve this problem faster, it is useful to solve gathering in the presence of Byzantine robots faster. Second, it would be good to know if finding a map is necessary in order for robots to settle. If it can be bypassed, then perhaps faster algorithms for when robots are initially gathered can be found. Finally, it would be interesting to find ways to increase the Byzantine tolerance of algorithms for general graphs. We do have one algorithm that handles $n-1$ Byzantine robots, but that algorithm only works for certain graphs.

\section*{Acknowledgments}
\label{sec:acknowledgments}

The authors would like to thank Gopal Pandurangan for fruitful discussions.

\bibliographystyle{plain}
\bibliography{ref}

\end{document}